\documentclass[pdflatex,sn-mathphys-num]{sn-jnl}
\usepackage{graphicx}%
\usepackage{multirow}%
\usepackage{amsmath,amssymb,amsfonts}%
\usepackage{amsthm}%
\usepackage{mathrsfs}%
\usepackage[title]{appendix}%
\usepackage{xcolor}%
\usepackage{textcomp}%
\usepackage{manyfoot}%
\usepackage{booktabs}%
\usepackage{algorithm}%
\usepackage{algorithmicx}%
\usepackage{algpseudocode}%
\usepackage{listings}%
\usepackage{enumitem}%
\usepackage{bm}
\usepackage{enumitem}

\usepackage{mathtools} 
\usepackage{booktabs} 
\usepackage{caption}  
\usepackage{array}    
\usepackage{float}  
\theoremstyle{theorem}%
\newtheorem{theorem}{Theorem}
\newtheorem{lemma}{Lemma}
\newtheorem{proposition}{Proposition}
\newtheorem{corollary}[theorem]{Corollary}

\newtheorem{case}{Case}

\theoremstyle{definition}
\newtheorem{definition}{Definition}
\newtheorem{example}{Example}
\newtheorem{remark}[theorem]{Remark}
\theoremstyle{construction}
\newtheorem{construction}{Construction}
\newtheorem{open problem}{Open problem}
\begin{document}
	
	\title[Capacity-Achieving Codes for Noisy Insertion Channels]{Capacity-Achieving Codes for Noisy Insertion Channels}

	\author[1]{\fnm{Hengfeng} \sur{Liu}}\email{hengfengliu@163.com}
	
	\author*[2]{\fnm{Chunming} \sur{Tang}}\email{tangchunmingmath@163.com}
	
	\author[3]{\fnm{Cuiling} \sur{Fan}}\email{cuilingfan@163.com}

	\affil[1,3]{\orgdiv{School of Mathematics}, \orgname{Southwest Jiaotong University}, \city{Chengdu}, \postcode{611756},  \country{China}}
	
	\affil[2]{\orgdiv{School of Information Science and Technology}, \orgname{Southwest Jiaotong University}, \city{Chengdu}, \postcode{611756}, \country{China}}


	\abstract{
		The design of error-correcting codes capable of correcting insertions and deletions has garnered significant attention recently, largely motivated by the requirements of DNA-based data storage. Insertion is one of the most frequent errors occurring in DNA sequences, where the inserted symbol is often identical or complementary to the original, and in practical implementations, noise can further cause the inserted symbol to mutate into a random one, which causes challenges to reliable data recovery. Motivated by these error mechanisms, this paper formalizes a noisy insertion channel characterized by an arbitrary number of identical or complementary symbol insertions, alongside at most one random insertion. Specifically, the exact coding capacity of this noisy insertion channel is established. By constructing asymptotically optimal error-correcting codes, this theoretical capacity is proven to be achievable. Furthermore, an efficient decoding algorithm is presented, which uniquely recovers the transmitted codewords in linear time with respect to the length of the received sequence.}

	\maketitle
	\section{Introduction}

The enormous expansion of data creates serious problems for conventional data storage media \cite{Reinsel2020}. In response to these issues, DNA storage has emerged as a viable substitute for next-generation data storage thanks to recent developments in DNA synthesis and sequencing technology. DNA storage offers unparalleled advantages over traditional electronic media, including six orders of magnitude higher data density, exceptional longevity, and the ability to generate copies efficiently. The feasibility of data storage in DNA molecules in-vitro (that is, outside of living cells and organisms) was initially demonstrated through experiments in \cite{cg, gb}, and later in-vivo (that is, within living cells and organisms) \cite{sn}. Moreover, in-vivo DNA storage enables critical biological functionalities such as watermarking genetically modified organisms, tagging infectious bacteria for epidemiological studies, conducting biogenetical research, and embedding computational memory for synthetic-biology applications \cite{Jupiter}. Similar to other storage systems, the transmitted information can be distorted by the channel, resulting in errors at the receiver end. In in-vivo DNA storage systems, information is corrupted by a variety of errors during different stages of the preservation process. Typical errors include point insertions/deletions, substitutions, which also commonly occur in electronic storage and communication systems. However, some errors are specific to DNA storage, for example, duplication, a special kind of burst insertion (insertion of consecutive bits). Generally, a duplication error occurs in a DNA sequence, a potentially modified copy of a substring is generated and inserted after the original substring.  \par 

 Another type of error unique in-vivo DNA storage is the complement insertion. A DNA sequence is composed of elements from the alphabet $ \{A, C, G, T\} $,  the four chemical bases: adenine, cytosine, guanine, and thymine. The four chemical bases are partitioned into complement pairs, where $A$ and $T$ are called \emph{complements} of each other, and so are $C$ and $G$. Besides random insertions, due to mutations during the biological processes in long-term evolution \cite{mh,sr}, DNA sequences are prone to point insertion of symbols complement or identical to the original one, where the former is called \emph{complement insertion} and the latter is \textit{tandem duplication}. For example, $TACTCTACCAA\Longrightarrow TAC\underline{G}TCTA\underline{A}CCAA$  demonstrates one complement insertion and one 1-tandem duplication. These errors pose the task of designing error-correcting codes. Error-correcting codes for insertions/deletions have been extensively studied since the first investigation by Varshamov, Tenengolts, and Levenshtein in the 1960s. In 1965, Varshamov and Tenengolts \cite{1965VT} constructed the famous binary VT codes correcting asymmetric errors on the Z-channel, and Levenshtein \cite{vi} subsequently proved that the VT codes can also correct a single insertion or deletion. Over the years, a lot of code constructions against insertions have been proposed (see \cite{Helberg, lius,Gabrys,Sima, lz} and references therein), while it is still a challenging task to correct a large number of insertions. In practical DNA storage implementations, insertion errors are often accompanied by biological noise \cite{Pumpernik2008}, leading to random substitutions or insertions. 
 
 Recently, Yohananov and Schwartz \cite{ys2} studied the \emph{exact} complement insertion channel and proposed optimal codes capable of correcting any number of complement insertions. However, their channel model assumes that insertions occur without any random noise. Motivated by the research in \cite{ys2} and the presence of biochemical noise, we investigate the \emph{noisy} insertion channel, where an arbitrary number of complement insertions and 1-tandem duplications, as well as up to one random insertion, may occur. This model extends the exact complement insertion channel studied in \cite{ys2} to a more realistic setting.
 
 By the \emph{coding capacity} of a channel, we mean the asymptotic rate of optimal error-correcting codes. In this work, we establish the coding capacity, denoted by $\mathsf{cap}^{\text{noisy}}_{q}$, for the proposed noisy insertion channel. By constructing asymptotically optimal codes, we prove that for an alphabet of size $q \ge 4$, the coding capacity is exactly
 $$\mathsf{cap}^{\text{noisy}}_{q} = \log_{q}(q-2).$$
 
 The major contributions of this work are summarized as follows:
 \begin{itemize}
 	\item \textbf{New Channel Modeling with Diverse Errors:} 
 	We extend the exact complement insertion channel in \cite{ys2} to the noisy insertion channel. This new model captures the occurrence of various insertion errors induced by biological noise in DNA storage.
 	
 	\item \textbf{Coding Capacity Determination:} 
 	We determine the asymptotic coding capacity $\mathsf{cap}^{\text{noisy}}_{q}$ for this noisy channel and prove that $\mathsf{cap}^{\text{noisy}}_{q} = \log_{q}(q-2)$ for $q \ge 4$. This shows that correcting additional random insertions and 1-tandem duplications incurs no asymptotic rate penalty compared to channels with only exact complement insertions.
 	
 	\item \textbf{Code Construction:} 
 	We construct error-correcting codes for the noisy insertion channel with rates asymptotically achieving $\log_{q}(q-2)$. These codes generalize the constructions in \cite{ys2} to correct additional error types.
 	
 	\item \textbf{Linear-Time Decoding Algorithm:} 
 	We present an explicit decoding algorithm that uniquely recovers the transmitted codewords in strictly linear time $O(N)$ with respect to the length of the received sequence.
 \end{itemize} 
 
 The rest of this paper is organized as follows. Section \ref{sec:2} introduces the relevant definitions and notations. Section \ref{sec:3} presents the asymptotically optimal code construction, determines its coding capacity, and details the decoding algorithm. Section \ref{conclusion} concludes the paper and outlines potential directions for future research.

	\section{Preliminaries} \label{sec:2}
	Throughout the paper, $\mathbb{Z} _{q}$ denotes the ring of integers modulo $q$, where $q\geq 2$ is a positive integer. For $n\in \mathbb{N}$, let $\mathbb{Z}_q^n$ denote the set of all sequences of length $n$ over $\mathbb{Z}_q$. 
	The set of all sequences of finite length over $\mathbb{Z}_q$ is denoted by $\mathbb{Z}_q^{*} $ and is defined by
	\[
	\mathbb{Z}_q^{*} = \bigcup_{n \geq 1} \mathbb{Z}_q^n.
	\]
	For a set $S$, $|S|$ denotes its size, while for a sequence $x$, $|x|$ denotes its length. Recall that the quaternary alphabet $ \{A, C, G, T\} $ is partitioned into two pairs of complement elements, where $A$ and $T$ are called \emph{complements} of each other, and so are $C$ and $G$. More generally, we have the following complement rule over $\mathbb{Z} _{q}$, where the alphabet size $q$ is supposed to be even.
	\begin{definition}
		A \textit{complement} operation on $\mathbb{Z} _{q}$ is a bijective map from $\mathbb{Z} _{q} \longrightarrow \mathbb{Z} _{q}$ defined by $u \mapsto \overline{u}$ such that $\overline{u} \ne u$ and $\overline{\overline{u}}=u$. 
	\end{definition}
	Throughout this paper, for convenience, we further assume $\overline{u}=q-1-u$. We naturally extend the complement notation to strings. Specifically, if 
	$x = x_0x_1\ldots x_{n-1} \in \mathbb{Z}_q^n$, then its complement is defined as 
	$\bar{x} = \bar{x}_0 \bar{x}_1 \ldots \bar{x}_{n-1}$. Given a sequence $x\in \mathbb{Z}_{q}^{*}$, a complement insertion generates a copy $\overline{v}$ which is complement to $v$ at $(i+1)$-th position in $x$, and inserts it immediately after $v$. More precisely, a \textit{complement insertion} is a function $T^{c}_{i} : \mathbb{Z}_{q}^{*} \longrightarrow \mathbb{Z}_{q}^{*}$ defined by
		\begin{equation}
		T^{c}_{i}(x)=
		\begin{cases} 
			uv\overline{v}w & \text{if } x =uvw , |u| = i, |v| = 1 \\ 
			x & \text{if } |x| < i + 1
		\end{cases}	
		\end{equation}		
	Denote by $\mathcal{T}= \left \{T^{c}_{i} \mid i\ge0 \right \}   $  the set of all complement insertion rules. We say that $y$ is a \emph{C-descendant} of $x$ if there exist $t \geq 0$ and $T^{c}_{i_{j}}\in \mathcal{T}  $ for $1\le j \le t$, such that
	$$y=T^{c}_{i_{t}}(T^{c}_{i_{t-1}}\cdots (T^{c}_{i_{1}}(x))).$$
	We define  \textit{1-tandem duplication} rule $T^{t}_{i} : \mathbb{Z}_{q}^{*} \longrightarrow \mathbb{Z}_{q}^{*}$, as
	\begin{equation}
	T^{t}_{i}(x)=
	\begin{cases} 
		uvvw & \text{if } x =uvw , |u| = i, |v| = 1 \\ 
		x & \text{if } |x| < i + 1
	\end{cases}	
\end{equation}
		Analogously, $y$ is said to be a \emph{T-descendant} of $x$ if it is obtained by only 1-tandem duplications. Further, if $y$ is obtained through complement insertions and 1-tandem duplications, then it is called a \emph{CT-descendant} of $x$.
		
		The following example illustrates the CT-descendant of a sequence.
	\begin{example}\label{example RC}
		Consider $\mathbb{Z}_4=\{0, 1, 2, 3\}$ with $\overline{0}=3$, $\overline{1}=2$. We have
		$$x=100231020\Longrightarrow 10023\underline{\textcolor{blue}{0}}102\underline{\textcolor{blue}{1}}0\Longrightarrow y= 10023\underline{\textcolor{blue}{0}}102\underline{\textcolor{blue}{1}}\underline{\textcolor{gray}{1}}0.$$ Here $y$ is a CT-descendant of $x$ obtained by two complement insertions and one 1-tandem duplication.
	\end{example}
In noisy channels, when a complement insertion occurs, the generated copies may not be complement or identical, and they always suffer from substitution errors. Further, the inserted complement symbol can by substituted by a random symbol from the alphabet, which is equivalent to a \textit{random insertion}. In this paper, we shall limit our attention to the noisy complement insertion channel where up to one random insertion is allowed.\par 
Generalizing the concept of \emph{CT-descendant}, we also have the concept of \emph{noisy descendant}. In a noisy channel, $y$ is called a \emph{noisy descendant} of $x$ if it is obtained through complement insertions, 1-tandem duplications and up to one random insertions. Further, the \emph{C-descendant cone} (resp., \emph{T-descendant cone}) of $x$ is the set of all its \emph{C-descendants} (resp., \emph{T-descendants}), denoted by $D^{*}_{c}(x)$ (resp., $D^{*}_{t}(x)$). The set of all  CT-descendants of $x$ is denoted by $D^{*}_{ct}(x)$. For $t\ge p\ge 0$, let $D^{*(1)}_{ct}(x)$ be the set obtained from $x$ by many complement insertions, 1-tandem duplications and one random insertion. Then we define the \emph{noisy descendant cone} of $x$ to be the set obtained by many complement insertions, 1-tandem duplications and at most 1 random insertion, formally written as 
\begin{equation}\label{noisy descendant}
	D^{*(\leq 1)}_{ct}(x) = D^{*}_{ct}(x) \cup D^{*(1)}_{ct}(x).
\end{equation}
 Continuing Example \ref{example RC}, we provide the following example in the noisy channel.
	\begin{example}\label{example NRC}
	Consider $\mathbb{Z}_4=\{0, 1, 2, 3\}$ with $\overline{0}=3$, $\overline{1}=2$. Here the sequence $x$ first suffers from a complement insertion of symbol 0, and then a random insertion of symbol 0, which is followed by two 1-tandem duplications, thus  $y'\in D^{*(\leq 1)}(x)$. 
	$$x=100231020\Longrightarrow 10023\underline{\textcolor{blue}{0}}1020\Longrightarrow 10023\underline{\textcolor{blue}{0}}102\underline{\textcolor{red}{0}}0\Longrightarrow y'=1002\underline{\textcolor{gray}{2}}3\underline{\textcolor{blue}{0}}102\underline{\textcolor{red}{0}}\underline{\textcolor{gray}{0}}0.$$
\end{example}
Next, we introduce definitions relevant to error-correcting codes. We first give the standard definition of error-correcting code for \emph{exact} complement insertion channel.
\begin{definition}
	A subset $\mathcal{C} \subseteq\mathbb{Z} _{q}^{n}$ is said to be a \textit{complement insertion-correcting code} if for any $x, y\in \mathcal{C}$ with $x\ne y$, \begin{equation}\label{error ball}
		D^{*}_{c}(x)\cap D^{*}_{c}(y)=\emptyset.
	\end{equation}
	
\end{definition}
From Eq. \eqref{noisy descendant}, it is obvious that $ D_{c}^{*}(x)\subseteq D_{ct}^{*(\le 1)}(x) $ for any sequence $x\in \mathbb{Z}_q^{*}$. Similarly we have the following definition for noisy insertion-correcting code.
\begin{definition}
	A subset $\mathcal{C} \subseteq\mathbb{Z} _{q}^{n}$ is said to be a \textit{code correcting infinitely many insertions of complement or identical symbol and up to one random insertion} if for any $x, y\in \mathcal{C}$ with $x\ne y$, \begin{equation}\label{noisy error ball}
		D_{ct}^{*(\le 1)}(x)\cap D_{ct}^{*(\le 1)}(y)=\emptyset.
	\end{equation} 
\end{definition}
For evaluation of above error-correcting codes, we consider their rates, defined as follows.
	\begin{definition}
		Let $\mathcal{C} \subseteq\mathbb{Z} _{q}^{n}$ be a $t$-error-correcting code of length $n$ and size $M$. Then the \textit{rate} of $\mathcal{C}$ is defined as 
		\begin{equation}
			R(\mathcal{C}) = \frac{1}{n} \log_{q}M.
		\end{equation}
		Let $A_q(n; t)$ denote the largest size of such $t$-error-correcting codes. The \textit{coding capacity} is defined as 
		\begin{equation}\label{CodCapacity}
			\mathsf{cap}_{q}=\limsup_{n \to \infty}
			  \frac{1}{n} \log_{q} A_q(n; t).
		\end{equation}
	\end{definition}
Throughout this paper, we say an error-correcting code is \emph{optimal} if it has the maximum size (equivalently, maximum rate). As $n$ goes to infinity, if the rate of a code equals to that of optimal codes, then the code is called \emph{asymptotically optimal}. 
Finally, the \emph{coding capacity} of a channel is the asymptotic rate of optimal error-correcting code. We denote by the coding capacity of exact  complement (resp., noisy) insertion channel by $\mathsf{cap}^{\text{exact}}_{q}$ (resp., $\mathsf{cap}^{\text{noisy}}_{q}$).
\section{Codes for the Noisy Insertion Channel}\label{sec:3}
In this section, we consider the noisy insertion channel, where any number of complement insertions, 1-tandem duplications and up to one random insertion may occur.\par 

By constructing codes that asymptotically achieve the coding capacity of \emph{exact} complement channels, we prove that correcting extra errors does not lower the channel's coding capacity. According to the definition, error-correcting codes, in both the \textit{exact} complement and \textit{noisy} insertion channels, can be constructed by identifying a subset that satisfies the condition that the two different descendant cones are disjoint, corresponding to each pair of two distinct elements of the set. Note that $D_{c}^{*}(x) \subseteq D_{ct}^{*(\leq 1)}(x)$ for any sequence $x$. Hence, every code designed for the \textit{noisy} insertion channel also serves as a code for the \textit{exact} complement insertion channel. \par 

In \cite{ys2}, the authors presented a key idea for the characterization of two different C-descendant cones to be disjoint for the exact complement insertion channel (see Eq.~(\ref{error ball})), where it was referred to as \emph{signature}. For a sequence $x\in \mathbb{Z}_{q}^{*}$, its signature is defined as follows.
 \begin{definition}
 Let $a\in \mathbb{Z}_{q}$, and let $a^{\oplus}$ denote the set of all sequences starting with $a$ and followed by any finite
 length string composed of $a$ and $\overline{a}$ only. For a given sequence $x\in \mathbb{Z}_{q}^{*}$, assume $x\in a_{0}^{\oplus}a_{1}^{\oplus}\ldots a_{l-1}^{\oplus}$, where $a_{i}\notin \left \{ a_{i+1},\overline{a_{i+1}} \right \} $ for $0\leq i \leq l-1$. Then we define the \emph{signature} of $x$ to be $$\sigma (x) = a_{0}a_{1}\ldots a_{l-1}.$$ 
 
 Moreover, the sequence $x$ is called \emph{irreducible} if $\sigma (x)=x$. We denote the set of all irreducible sequences of length $n$ by  $\mathrm{Irr}(n)$.
 \end{definition}
 \begin{example}
  Consider $\mathbb{Z}_4=\{0, 1, 2, 3\}$ with $\overline{0}=3$, $\overline{1}=2$. Let $x=0312130$ and  $y=1320102$ be two sequences over $\mathbb{Z}_4$. We have $\sigma (x)=013$ and $\sigma (y)=y$.  Here $y$ is irreducible.
 \end{example}
 
The following result provides a necessary and sufficient condition for a finite sequence over $\mathbb{Z}_q$ to be irreducible.
 \begin{theorem}
 	Let $x=x_0x_1\dots x_{n-1}\in \mathbb{Z}_q^{*}$ be an arbitrary sequence. Then $x$ is irreducible if and only if  $x_i\notin \{x_{i+1}, \overline{x_{i+1}} \} $ for all $i=0,1, \dots, n-1$.
 \end{theorem}
 \begin{proof}
 	Proof follows directly from the definition.
 \end{proof}
 
 The following proposition demonstrates that insertions of complement or identical symbols do not alter the signature of a sequence.
  \begin{proposition}\label{do not alter}
 	Let $x, y\in \mathbb{Z}_{q}^{*}$. If $y\in D_{ct}^{*}(x)$, then $\sigma (x)=\sigma (y)$.
 \end{proposition}
\begin{proof}
	Let $x \in \mathbb{Z}_q^*$, and let its signature be $\sigma(x) = a_0 a_1 \dots a_{l-1}$ where $a_{i+1} \notin \{a_i, \overline{a_i}\}$. Further, we partition $x$ into $l$ blocks and denote $x = B_0 B_1 \dots B_{l-1}$, where $B_i \in a_i^\oplus$ is a substring of $x$ starting with $a_i$. Let $y' \in D_{ct}^*(x)$. Without loss of generality, assume that $y'$ is a sequence obtained from $x$ via a single complement insertion or a 1-tandem duplication. Suppose this insertion occurs in the block $B_i$, then the effect of the insertion is merely extending the substring $B_i$ by one letter, while its first letter $a_i$ remains unchanged. Thus, $\sigma(x) = \sigma(y')$. By applying this argument repeatedly, we have $\sigma(x) = \sigma(y)$ for any $y \in D_{ct}^*(x)$.
\end{proof}
 In \cite[Theorem 5]{ys2}, the authors provided a characterization for two distinct sequences to have a common C-descendant cone. Now, we extend this result for CT-descendant cone. The following theorem presents a characterization for two distinct sequences over $\mathbb{Z}_q$ to have a common CT-descendant cone.
 \begin{theorem}\label{signature characterizes}
Let $x\ne y\in \mathbb{Z}_{q}^{*}$. Then  $ D_{ct}^{*}(x) \cap D_{ct}^{*}(y) \ne \emptyset $ if and only if $\sigma (x)= \sigma (y)$.
 \end{theorem}
 \begin{proof}
 Let $z \in  D_{ct}^{*}(x) \cap D_{ct}^{*}(y) \ne \emptyset $. Then by Proposition \ref{do not alter}, we have
 \begin{equation*}
 	\sigma (x)= \sigma (z)=\sigma (y).
 \end{equation*}
 For the other direction, let $\sigma (x)= \sigma (y)$. Then by using \cite[Theorem 5]{ys2}, we have 
 \begin{equation*}
 	 D_{c}^{*}(x) \cap D_{c}^{*}(y) \ne \emptyset,
 \end{equation*} 
 which follows that $ D_{ct}^{*}(x) \cap D_{ct}^{*}(y) \ne \emptyset $.
 \end{proof}

Complement insertions and 1-tandem duplications do not alter the signature of a sequence, but in the \textit{noisy} insertion channel, the signature might change. In our construction, for convenience, we restrict the codebook to $\mathrm{Irr}(n)$, the set of irreducible sequences of length $n$. We will see this restriction does not decrease the asymptotic rate.\par 
By using constraint coding approaches and investigating how the signatures change in the \textit{noisy} channel, we construct codes that enable the decoder to recover the signatures (i.e., irreducible codewords) in a unique way. In the \textit{noisy} channel, for a given sequence $x\in \mathbb{Z}_{q}^{n}$, its descendant $x''\in D_{ct}^{*(\le 1)}(x)$ can be generated in the following three ways:
 \begin{enumerate}
 \item complement insertions and 1-tandem duplications;
 \item  many complement insertions and 1-tandem duplications followed by a random insertion;
 \item  complement insertions and 1-tandem duplications followed a random insertion, which followed by complement insertions and 1-tandem duplications.
 \end{enumerate}
 
By Proposition~\ref{do not alter}, complement insertions and 1-tandem duplications do not alter the signature of a sequence; hence, only case 2 needs to be considered for $\sigma(x'')$. Further, in case $2$, we assume that the random insertion occurs in $x'$, which is obtained from $x$ by many complement insertions and 1-tandem duplications. Since we restrict the codeword $x$ to the set of irreducible sequences $\mathrm{Irr}(n)$, then $$\sigma(x')=\sigma(x)=x.$$
 Therefore, we aim at figuring out how does a random insertion alters the the signature of $x'$, and then design corresponding error-correcting code (ECC) to correct the errors in signatures to recover $\sigma(x')$, illustrated as follows.
$$x=\sigma(x')
\mathrel{
	\substack{
		\xrightarrow{\text{\fontsize{9}{10}\selectfont signature altered}} \\[0.5ex]
		\xleftarrow[\text{\fontsize{9}{10}\selectfont recover via ECC}]{\hspace{2em}}
	}
}
\sigma(x'').$$
\begin{lemma}\label{a noisy alter}
Suppose $x\in \mathrm{Irr}(n)$ and $x'\in D_{ct}^{*}(x)$. Let a random insertion occur in $x'$ and the resulting new sequence be $x''$. Then there are the following three possible ways:
\begin{itemize}
	\item[(i)] a point insertion,
	\item[(ii)] a point substitution by a complementary symbol,
	\item [(iii)] a burst insertion of length 2,
\end{itemize}
for the signature $x=\sigma(x')$ to be changed into $\sigma(x'')$.
\end{lemma}
\begin{proof}
	Assume $x=x_{1}x_{2}\ldots x_{n}\in \mathrm{Irr}(n)$, then $x'\in x_{1}^{\oplus}x_{2}^{\oplus}\ldots x_{n}^{\oplus}$. Let $x'= B_{1}B_{2}\ldots B_{n} $, where $B_{i}\in x_{i}^{\oplus}$ for $1\leq i \leq n$. Suppose the random insertion occurs in $B_{i} = x_{i} b_{i,1}\ldots b_{i,k}$, where $b_{i,l}\in \left \{ x_{i},\overline{ x_{i}} \right \}  $, $1\le l\le k$, and the inserted symbol is denoted by $r\in \mathbb{Z}_{q}$. The way of altering the signature depends on the location at which the insertion occurs in $B_{i}$. Now we consider the following two cases:
	\begin{case} 
		If the point insertion occurs at the end of $B_{i}$, that is, 
		$$x'=B_{1}\ldots B_{i-1}{\color{blue}x_{i}b_{i,1}\ldots b_{i,k} }B_{i+1}\ldots B_{n} $$ $$\Downarrow$$ $$x''=B_{1}\ldots B_{i-1}{\color{blue}x_{i}b_{i,1}\ldots b_{i,k}{\color{red}r} }B_{i+1}\ldots B_{n}$$
		
		\begin{itemize}
			\item[(i)] If $r\in \left \{ x_{i},\overline{ x_{i}} \right \}$ or $r=x_{i+1}$, then $\sigma(x'')=\sigma(x')=x$, the signature remains unchanged.
			\item [(ii)] If $r=\overline{x_{i+1}}$, the signature  changes from $x=x_{1}\ldots x_{i} x_{i+1}x_{i+2}\ldots x_{n}$ to $\sigma(x'')=x_{1}\ldots x_{i} \overline{x_{i+1}}x_{i+2}\ldots x_{n}$, where a substitution of a complementary symbol happens.
			\item[(iii)] If $r\notin  \left \{ x_{i},\overline{ x_{i}}, x_{i+1},  \overline{ x_{i+1}}\right \} $, the signature changes from $x=x_{1}\ldots x_{i} x_{i+1}\ldots x_{n}$ to $\sigma(x'')=x_{1}\ldots x_{i} {\color{red}r}x_{i+1}\ldots x_{n}$, where a point insertion happens.
		\end{itemize}
	\end{case}
	\begin{case}
		If the point insertion does not occur at the end of $B_{i}$, denote $x_{i} = b_{i,0}$, then there exist some $l$ ($0\le l\le k-1$), such that $$x'=B_{1}\ldots B_{i-1}{\color{blue}x_{i}b_{i,1}\ldots b_{i,l}b_{i,l+1} \ldots b_{i,k} }B_{i+1}\ldots B_{n}$$ $$\Downarrow$$ $$x''=B_{1}\ldots B_{i-1}{\color{blue}x_{i}b_{i,1}\ldots b_{i,l}{\color{red}r}b_{i,l+1} \ldots b_{i,k} }B_{i+1}\ldots B_{n}.$$
		\begin{itemize}
			\item[(i)] $r\in \left \{ x_{i},\overline{ x_{i}} \right \}$, then the signature remains unchanged.
			\item [(ii)] $r\notin \left \{ x_{i},\overline{ x_{i}} \right \}$, the signature changes from $x=x_{1}\ldots x_{i} x_{i+1}\ldots x_{n}$ to $\sigma(x'')=x_{1}\ldots x_{i} {\color{red}rb_{i,l+1}}x_{i+1}\ldots x_{n}$, where $b_{i,l+1}\in  \left \{ x_{i},\overline{ x_{i}} \right \}$. This is a burst insertion of length 2.
		\end{itemize}
	\end{case}
\end{proof}

In order to correct the three types of errors that occur in signatures, we shall further impose corresponding constraints on $\mathrm{Irr}(n)$. We first address the case of a point substitution by a complementary symbol. 

Now, we present the following code construction, which can correct the above error and plays a crucial role in our final code construction.
\begin{construction}\label{correct a complementary symbol}
Given integers $0\le a \le 2q-1$ and $0\le b\le qn-1$, we construct the code $C_{(a,b)}(n)$ as 
\begin{align}\label{Const1}
	C_{(a,b)}(n) = \left \{ u\in \mathbb{Z}_{q}^{n}\mid\sum_{i=1}^{n}u_{i} \equiv a \pmod {2q} ,\quad \sum_{i=1}^{n}iu_{i} \equiv b\pmod{qn}  \right \}. 
\end{align}
\end{construction}

\begin{theorem}
The code $C_{(a,b)}(n)$ defined by Eq. \eqref{Const1}, is able to correct a point substitution by a complementary symbol.
\end{theorem}
\begin{proof}
When a point substitution by a complementary symbol occurs in $u \in \mathbb{Z}_{q}^{n}$, let the symbol $w$ at the $i$-th position be replaced by its complement $q-1-w$. We prove that the pair $(w,i)$ can be uniquely determined, thus $u$ is uniquely recovered. Let $u''$ be obtained from $u$, then we have
\begin{equation}\label{1chan}
\sum_{i=1}^{n}u''_{i}-\sum_{i=1}^{n}u_{i}=q-1-2w,
\end{equation}

\begin{equation}\label{2chan}
	 \sum_{i=1}^{n}iu''_{i}- \sum_{i=1}^{n}iu_{i}=i(q-1-2w).
\end{equation}
If there exists another sequence $u' \in \mathbb{Z}_{q}^{n}$ with a pair $(w',i')$ such that $u''$ is obtained by substituting $w'$ with its complement at position $i'$, then it follows from Eqs.~(\ref{1chan}) and~(\ref{2chan}) that
\begin{equation}\label{1chanmod}
q-1-2w \equiv q-1-2w'\pmod{2q},
\end{equation}

\begin{equation}\label{2chanmod}
	i(q-1-2w) \equiv i'(q-1-2w')\pmod{qn}.
\end{equation}
By Eq. (\ref{1chanmod}), $q$ divides $w-w'$, which means $w=w'$. By Eq. (\ref{2chanmod}), we have 
\begin{equation}
qn\mid	(i-i')(q-1-2w). 
\end{equation}
However, $\left |  (i-i')(q-1-2w) \right |  \le (q-1)(n-1)<qn$, then $i=i'$, which completes the proof.
\end{proof}
\begin{corollary}\label{cab}
For any $q$ and $n$, there exist some $a$ and $b$ such that $$|\mathcal{C}_{(a,b)}(n)|\ge \displaystyle\frac{q^{n-2}}{2n}.$$
\end{corollary}
\begin{proof}
The code family $C_{(a,b)}(n)$ ($0\le a\le 2q-1$, $0\le b\le qn-1$)  $\mathbb{Z}_{q}^{n}$ form a partition of $\mathbb{Z}_{q}^{n}$, then the lower bound follows from the pigeonhole principle.
\end{proof}
\begin{example}
In this example, let $q=4$, we list all codewords for the code $\mathcal{C}_{(0,0)}(6)$.

	$\mathcal{C}_{(0,0)}(6)$ $$=
 \left\{
\begin{aligned}
	&(0, 0, 0, 0, 0, 0),\ (0, 3, 2, 3, 0, 0),\ (0, 3, 3, 1, 1, 0),\ (1, 1, 3, 3, 0, 0),\ (1, 2, 2, 2, 1, 0), \\
	&(1, 2, 3, 0, 2, 0),\ (1, 2, 3, 1, 0, 1),\ (1, 3, 0, 3, 1, 0),\ (1, 3, 1, 1, 2, 0),\ (1, 3, 1, 2, 0, 1), \\
	&(1, 3, 2, 0, 1, 1),\ (2, 0, 3, 2, 1, 0),\ (2, 1, 1, 3, 1, 0),\ (2, 1, 2, 1, 2, 0),\ (2, 1, 2, 2, 0, 1), \\
	&(2, 1, 3, 0, 1, 1),\ (2, 2, 0, 2, 2, 0),\ (2, 2, 0, 3, 0, 1),\ (2, 2, 1, 0, 3, 0),\ (2, 2, 1, 1, 1, 1), \\
	&(2, 2, 2, 0, 0, 2),\ (2, 3, 0, 0, 2, 1),\ (2, 3, 0, 1, 0, 2),\ (3, 0, 1, 2, 2, 0),\ (3, 0, 1, 3, 0, 1), \\
	&(3, 0, 2, 0, 3, 0),\ (3, 0, 2, 1, 1, 1),\ (3, 0, 3, 0, 0, 2),\ (3, 1, 0, 1, 3, 0),\ (3, 1, 0, 2, 1, 1), \\
	&(3, 1, 1, 0, 2, 1),\ (3, 1, 1, 1, 0, 2),\ (3, 2, 0, 0, 1, 2)
\end{aligned}
\right\}
.$$
There are $33$ codewords in total, while the lower bound in Corollary \ref{cab} is $\displaystyle\frac{4^4}{12}<22$.
\end{example}

To correct a point insertion, we use a slightly modified version of the $q$-ary Varshamov-Tenengolts (VT) code \cite{Tenengolts1984}, which is a non-binary generalization of the binary VT code \cite{vi} and corrects a point insertion or deletion. To meet the constraint in Construction \ref{correct a complementary symbol}, the modified version uses the congruency $2q$ instead of $q$ in the original code, as the following construction shows.
\begin{construction}\label{qVT}
Given integers $0\le c \le 2q-1$ and $0\le d \le n-1$, the following code is able to correct a point insertion.
\begin{align}
	C_{T(c,d)}(n) = \left \{ u\in \mathbb{Z}_{q}^{n}\mid\sum_{i=1}^{n}u_{i} \equiv c\pmod {2q} ,\quad \sum_{i=1}^{n}(i-1)\beta _{i} \equiv d\pmod{n}  \right \} ,
\end{align}
where $\beta_{1}=1$ and for $2 \le i \le n$, \[
\beta_{i} = 
\begin{cases} 
	1 & \text{if } u_{i} \geq u_{i-1}, \\ 
	0 & \text{if } u_{i} < u_{i-1}. 
\end{cases}
\]
\end{construction}

By Lemma \ref{a noisy alter}, the random insertion may also cause a burst insertion of length 2 in signatures. Codes correcting a burst insertion have gained significant attention in recent years \cite{sw,Schoeny2,sn1, lz, wangshu 2023,song2023}, with the best known code proposed very recently by Sun~\textit{et~al.} \cite{sun2025}. Early in 2017, Schoeny~\textit{et~al.} \cite{sw} proved that codes can correct a burst insertions if and only if correct a burst deletions.\par 
To correct the burst insertion in the signature, we employ the combination of \emph{$P$-bounded single-deletion-correcting code} and\emph{ run-length limited (RLL) VT-code}, which was first formalized by  Schoeny~\textit{et~al.} \cite{sw}. We first cover some useful concept and lemma, and then propose a modified code construction (with different parameters) to meet the restriction that our codewords is chosen from $\mathrm{Irr}(n)$.
\begin{definition}
A set is called a $P$-\textit{bounded single-deletion-correcting code} if the decoder can correct a single deletion given knowledge of the location of the deleted symbol to within $P$ consecutive positions.
\end{definition}
We will employ the following $q$-ary shifted VT (SVT) code, which is a modiﬁed version of the $q$-ary VT-code.
\begin{construction} \label{SVT code}Given $0\le e\le P$, $0\le f\le q-1$ and $g\in \left\{0,1\right\}$, the $q$-ary shifted VT code $SVT_{(e,f,g)}(n,P,q)$ is defined to be 
	\begin{equation}\label{shiftvt}
		\begin{split}
			SVT_{(e,f,g)}(n,P,q) = \bigg\{ x\in\mathbb{Z}_{q}^{n}\mid 
		\displaystyle\sum_{i=1}^{n} i\beta_{i}\equiv e \pmod {P+1},&\displaystyle\sum_{i=1}^{n} x_{i}\equiv f\pmod {q},\\
		&\displaystyle\sum_{i=1}^{n}\beta_{i}\equiv g\pmod {2} \bigg\}.
		\end{split}
	\end{equation}

where $\beta_{1}=1$ and for $2 \le i \le n$, \[
\beta_{i} = 
\begin{cases} 
	1 & \text{if } x_{i} \geq x_{i-1}, \\ 
	0 & \text{if } x_{i} < x_{i-1}. 
\end{cases}
\]
\end{construction}

Construction 3 modifies the standard $q$-ary VT code by replacing the modulo-$n$ condition with a modulo-$(P+1)$ constraint to reduce redundancy. The additional parity constraint $\sum_{i=1}^{n}\beta_{i}\equiv g \pmod{2}$ is introduced to compensate for the loss of positional information. As formally stated in the following lemma, this structure is sufficient to correct a single deletion when its location is known to within $P$ positions.

\begin{lemma}
For any $0\le e\le P$, $0\le f\le q-1$ and $g\in \left\{0,1\right\}$, the $q$-ary shifted VT code $SVT_{(e,f,g)}(n,P,q)$ is a $P$-bounded single-deletion-correcting code.
\end{lemma}
\begin{proof}
	Let $x \in SVT_{(e,f,g)}(n,P,q)$ be the transmitted codeword, and let $x'$ be the sequence of length $n-1$ received after a single symbol deletion. By the definition of a $P$-bounded single-deletion-correcting code, the decoder is provided with a window of $P$ consecutive indices, say $W = \{j, j+1, \dots, j+P-1\}$, indicating the potential locations of the deleted symbol in $x$.
	
	First, the value of the deleted symbol, denoted by $v$, can be uniquely determined using the second constraint of Construction 3. Since $\sum_{i=1}^n x_i \equiv f \pmod{q}$, we simply have $v \equiv f - \sum_{i=1}^{n-1} x'_i \pmod{q}$.
	
	Next, we must determine the exact position of the deletion. Let $\beta \in \{0,1\}^n$ and $\beta' \in \{0,1\}^{n-1}$ be the binary difference sequences (as defined in Eq. (\ref{shiftvt})) corresponding to $x$ and $x'$, respectively. Note that a single symbol deletion in $x$ at index $k$ ($k \in W$) translates to exactly one bit deletion in the binary sequence $\beta$. Specifically, depending on the values of the adjacent symbols, the deleted bit in $\beta$ is either $\beta_k$ or $\beta_{k+1}$. Since the spatial index $k$ is restricted to the $P$ positions in $W$, the location of the deleted bit in $\beta$ is strictly constrained to a window of at most $P+1$ positions: $\{j, j+1, \dots, j+P\}$.
	
	The first and third constraints of the code state that $\sum_{i=1}^n i\beta_i \equiv e \pmod{P+1}$ and $\sum_{i=1}^n \beta_i \equiv g \pmod{2}$. These two congruences precisely constitute a binary shifted VT code over $\beta$. It has been established in \cite{sw} that a binary shifted VT code with modulus $P+1$ is capable of correcting a single bit deletion provided the deletion location is known to within a window of $P+1$ positions. Consequently, the complete binary sequence $\beta$ can be uniquely recovered from $\beta'$.
	
	Once the binary difference sequence $\beta$ is recovered, the run structure (i.e., the lengths of all non-decreasing and strictly decreasing runs) of the original codeword $x$ is completely revealed. Following the standard decoding principles of the $q$-ary VT code, knowing both the value of the deleted symbol $v$ and the restored run structure from $\beta$ allows us to uniquely identify the exact run from which $v$ was deleted. Since reinserting $v$ into any position within its native run in $x'$ yields the identical sequence, the original codeword $x$ is uniquely recovered.
\end{proof}
To correct a burst deletion of length 2, we treat a codeword $x$ as a $2\times \frac{n}{2}$ codeword array $A_{2}(x)$, where the code length $n$ is supposed to be even. A codeword is put in the array transmitted column-by-column, so that when a 2-burst insertion happens, each row of $A_{2}(x)$ suffers from a point deletion:
$$A_{2}(x)=\begin{bmatrix}
	x_{1} &\ldots&x_{i}& \ldots &x_{n-1} \\
	x_{2}&\ldots&x_{i+1}&\ldots  &x_{n}
\end{bmatrix}.$$
Denote by $A_{2}(x)_{i}$ the $i$-th row of $A_{2}(x)$. Furthermore, when suffering a 2-burst insertion, suppose the $j$-th bit of first row $A_{2}(x)_{1}$ is deleted, then the position of the deleted bit in $A_{2}(x)_{2}$ is either $j$ or $j-1$. Therefore, we will use run-length limited (RLL) VT-code to encode the $A_{2}(x)_{1}$, which ensure that the location of the deletion in $A_{2}(x)_{1}$ can be determined within consecutive positions of the maximum run length. Then the code can correct a 2-burst deletion (insertion) once $A_{2}(x)_{2}$ is encoded by a suitably chosen $P$-bounded single-deletion-correcting code.
\begin{definition}
	A $q$-ary vector $x$ of length $n$ is said to satisfy the $f(n)$-RLL$(n,q)$ constraint, and is called an $f(n)$-RLL$(n,q)$ vector, if the length of its longest run is at most $f(n)$.
\end{definition}

Further, denote the set of all $f(n)$-RLL$(n, q)$ vectors by $S_{n}(f(n))$ and let 
$$S^{\mathrm{Irr}}_{n}(f(n)) = \left \{ x\in \mathbb{Z}_{q}^{n}\mid A_{2}(x)_{1}\in S_{n/2}(f(n)),x\in\text{Irr}(n)  \right \}.$$
We establish a lower bound on the size of the set $S^{\mathrm{Irr}}_{n}(f(n))$ in the following lemma.

\begin{lemma} Let $f(n)=3\log_{q}n+2$, we have the following lower bound:
	$$\left | S^{\mathrm{Irr}}_{n}(3\log_{q}(n)+2) \right | \ge q(q-2)^{n-1} \left(1-\frac{1}{2(q-2)n^{1/2}}\right).$$
\end{lemma}
\begin{proof}
	We use a probabilistic argument to derive the lower bound. Let $X_n$ be a random variable denoting the maximum run length of $A_2(x)_1$, where $x$ is chosen uniformly at random from $\mathrm{Irr}(n)$. We aim to compute an upper bound on the probability $P(X_n \ge 3\log_q n + 2)$.
	
	By the union bound, it is sufficient to evaluate the probability that a specific window of length $L = 3\log_q n + 2$ in $A_2(x)_1$ constitutes a run (i.e., all $L$ elements are identical). Let this probability be $P_e$.
	
	Note that $A_2(x)_1$ is formed by extracting alternate symbols from $x$ (e.g., $x_1, x_3, x_5, \dots$). Thus, a contiguous window of length $L$ in $A_2(x)_1$ corresponds to a contiguous sub-block of length $2L - 1$ in the original codeword $x$. Since $x$ is uniformly distributed over $\mathrm{Irr}(n)$, any such sub-block is uniformly distributed over the set of all valid irreducible sequences of length $2L - 1$. By the definition of irreducible sequences (Theorem 1), the total number of such sequences of length $2L - 1$ is $q(q-2)^{2L-2}$. This serves as the denominator of $P_e$.
	
	For this specific window in $A_2(x)_1$ to be a run, the $L$ odd-positioned symbols within the corresponding sub-block in $x$ must be identical, say equal to $\alpha \in \mathbb{Z}_q$. There are $q$ choices for $\alpha$. Meanwhile, the remaining $L-1$ even-positioned symbols in this sub-block must satisfy the irreducible constraint. Since each even-positioned symbol is sandwiched between two identical symbols $\alpha$, it can neither be $\alpha$ nor its complement $\overline{\alpha}$. Thus, there are exactly $q-2$ choices for each of the $L-1$ even-positioned symbols. This yields $q(q-2)^{L-1}$ valid configurations, which forms the numerator of $P_e$.
	
	Substituting the window length $L = 3\log_q n + 2$, we have $2L - 2 = 6\log_q n + 2$ and $L - 1 = 3\log_q n + 1$. Therefore, the probability $P_e$ is precisely given by:
	$$P_e = \frac{q(q-2)^{3\log_q n + 1}}{q(q-2)^{6\log_q n + 2}} = \frac{1}{(q-2)^{3\log_q n + 1}} = \frac{1}{q-2}\cdot\frac{1}{n^{3\log_q(q-2)}}$$
	
	Since the function $g(q) = \log_q(q-2)$ is strictly increasing for $q \ge 4$, we can apply the union bound over the $\frac{n}{2}$ possible starting positions of the window in $A_2(x)_1$:
	$$P(X_n \ge 3\log_q n + 2) \le \frac{n}{2} \cdot P_e = \frac{n}{2}\cdot\frac{1}{q-2}\cdot\frac{1}{n^{3\log_q(q-2)}} \le \frac{1}{2(q-2)n^{1/2}}$$
	
	Consequently, the lower bound on the size of the set $S_n^{\mathrm{Irr}}(3\log_q n + 2)$ is derived as:
	$$|S_n^{\mathrm{Irr}}(3\log_q n + 2)| = |\mathrm{Irr}(n)| \cdot (1 - P(X_n \ge 3\log_q n + 2))$$
	$$\ge q(q-2)^{n-1} \left(1 - \frac{1}{2(q-2)n^{1/2}}\right)$$
\end{proof}

Now, we restrict our codewords to the set $S^{\mathrm{Irr}}_{n}(3\log_{q}(n)+2)$, and encode $A_{2}(x)_{1}$ using the $q$-ary VT code $C_{T(h,w)}$ (see Construction~\ref{qVT}), and $A_{2}(x)_{2}$ using the SVT code $SVT_{e,f,g}(n,3\log_{q}(n)+3,q)$ (see Construction~\ref{SVT code}). Then the following resulting code can correct a 2-burst insertion or deletion.

\begin{construction}\label{Irr burst}
For given integers $0\le h \le 2q-1$, $0\le w \le n-1$, $0\le e\le 3\log_{q}(n)+3$, $0\le f\le q-1$ and $g\in \left\{0,1\right\}$, we construct the code
\begin{equation}
	\begin{split}
		C_{B(h,w,e,f,g)}(n) = \bigg\{  x \in S^{\mathrm{Irr}}_{n}&(3\log_{q}(n)+2) \mid A_{2}(x)_{1} \in C_{T(h,w)}(n/2), \\
		& A_{2}(x)_{2} \in SVT_{e,f,g}(n/2,3\log_{q}(n)+3,q) \bigg\}.
	\end{split}
\end{equation}

\end{construction}
\begin{lemma}
The code $C_{B\left(h,w,e,f,g\right)}(n)$ can correct a burst insertion of length 2.
\end{lemma}
The proof follows directly from the previous discussion, and an application of the pigeonhole principle immediately yields the following lower bound on the code size.
\begin{corollary}\label{C_{B}}
There exist parameters with $0\le h \le 2q-1$, $0\le w \le n-1$, $0\le e\le 3\log_{q}(n)+3$, $0\le f\le q-1$ and $g\in \left\{0,1\right\}$, such that 
\begin{align}
	\left |C_{B\left(h,w,e,f,g\right)}(n) \right | \ge \frac{q(q-2)^{n-1}}{2\left(3\log_{q}(n)+4\right)q^{2}n}\left(1-\frac{1}{2(q-2)n^{1/2}}\right).
\end{align}
\end{corollary} 

\begin{remark}
 Taking the base-$q$ logarithm of the lower bound in Corollary \ref{C_{B}} yields the redundancy:
	$$r(C_B) \le \log_q n + \log_q(3\log_q n + 4) + O(1).$$
	This demonstrates that the redundancy of $C_B$ is strictly bounded by $O(\log_q n)$. Consequently, the asymptotic rate penalty $\frac{r(C_B)}{n}$ vanishes.
\end{remark}

Combining Construction \ref{Irr burst} and constraints in Construction \ref{correct a complementary symbol} and Construction \ref{qVT} together, we obtain the following final code construction, which can correct any number of complement insertions, 1-tandem duplications and up to one random insertion.
\begin{construction}\label{final construction}
Let $n$ be an even number, for $0\le a \le 2q-1$, $0\le b\le qn-1$, $0\le d \le n-1$, $0\le h \le 2q-1$, $0\le w \le n-1$, $0\le e\le 3\log_{q}(n)+3$, $0\le f\le q-1$ and $g\in \left\{0,1\right\}$, we construct the code

\begin{equation}
\begin{split}
	\mathcal{C}_{F(a,b,d,h,w,e,f,g)}(n) = & \bigg\{x\in \mathbb{Z}_{q}^{n} \mid x\in C_{(a,b)}(n) \cap C_{T(a,d)}(n) \cap C_{B(h,w,e,f,g)}(n)\bigg\} \\
	= & \bigg\{x\in C_{B(h,w,e,f,g)}(n) \mid \sum_{i=1}^{n}x_{i} \equiv a \pmod{2q}, \\
	& \quad \sum_{i=1}^{n}ix_{i} \equiv b \pmod{qn}, \\
	& \quad \sum_{i=1}^{n}(i-1)\beta_{i} \equiv d \pmod{n}\bigg\},
\end{split}
\end{equation}
where $\beta_{1}=1$ and for $2 \le i \le n$, \[
\beta_{i} = 
\begin{cases} 
	1 & \text{if } x_{i} \geq x_{i-1}, \\ 
	0 & \text{if } x_{i} < x_{i-1}. 
\end{cases}
\]
\end{construction}

\begin{theorem}\label{main thm}
Error-correcting codes $\mathcal{C}_{F(a,b,d,h,w,e,f,g)}(n)$ in Construction \ref{final construction} are able to correct any number of complement insertions, 1-tandem duplications and up to one random insertion. Furthermore, there exist one such code with size 
\begin{equation}
\left | \mathcal{C}_{F(a,b,d,h,w,e,f,g)}(n) \right | \ge  \frac{q(q-2)^{n-1}}{4\left(3\log_{q}(n)+4\right)q^{3}n^{3}} \left(1-\frac{1}{2(q-2)n^{1/2}}\right)
\end{equation}
\end{theorem}
\begin{proof}
By the previous discussions, since signatures can be uniquely recovered, the decoder can correct any number of complement insertions, 1-tandem duplications, and up to one random insertion. Let $C_{B(h,w,e,f,g)}(n)$ be the code in Corollary \ref{C_{B}}. Then 
\begin{align}
	\left |C_{B\left(h,w,e,f,g\right)}(n) \right | \ge \frac{q(q-2)^{n-1}}{2\left(3\log_{q}(n)+4\right)q^{2}n}\left(1-\frac{1}{2(q-2)n^{1/2}}\right).
\end{align}
Further, the set $C_{B(h,w,e,f,g)}(n)$ is partitioned by the disjoint union of code families:
	$$C_{B(h,w,e,f,g)}(n)=\bigcup_{(a,b,d)\in\mathbb{Z}_{2q}\times \mathbb{Z}_{qn}\times \mathbb{Z}_{n}}\mathcal{C}_{F(a,b,d,h,w,e,f,g)}(n).$$ Then the lower bound is obtained using the pigeonhole principle.
\end{proof}
\begin{remark}
	To explicitly verify the asymptotic optimality of our construction, we evaluate the rate of the codes $\mathcal{C}_{F(a,b,d,h,w,e,f,g)}(n)$ utilizing the lower bound established in Theorem~\ref{main thm}. The code rate is defined as $R = \frac{1}{n}\log_q|\mathcal{C}_F|$. Taking the base-$q$ logarithm of the lower bound yields:
	\begin{align*}
		\log_q|\mathcal{C}_F| &\ge \log_q \left( (q-2)^{n-1} \right) - \log_q \left( n^3 (3\log_q n + 4) \right) + O(1) \\
		&= (n-1)\log_q(q-2) - 3\log_q n - \log_q(3\log_q n + 4) + O(1),
	\end{align*}
	where the $O(1)$ term encompasses all constant factors and $o(1)$ terms independent of the asymptotic growth of $n$. Dividing by the code length $n$ and taking the limit as $n \to \infty$, the logarithmic penalty terms $\frac{\log_q n}{n}$ and $\frac{\log_q(\log_q n)}{n}$ asymptotically vanish. Consequently, the achievable rate approaches:
	$$\lim_{n\to\infty} \frac{1}{n}\log_q|\mathcal{C}_F| \ge \log_q(q-2).$$
	
	Recall that $\mathsf{cap}^{\text{exact}}_{q} = \log_{q}(q-2)$ is precisely the theoretical capacity of the idealized exact complement insertion channel \cite{ys2}. Since a noisy channel incorporating additional random insertions cannot exceed the capacity of its noise-free counterpart, our achievable rate perfectly matches the theoretical upper bound. This direct calculation formally proves that equipping the code to correct additional random biological noise incurs an asymptotically zero rate penalty. Thus, the proposed codes are asymptotically optimal, establishing that $\mathsf{cap}^{\text{noisy}}_{q} = \log_q(q-2)$.
\end{remark}
Since our codes are also special cases of 1-tandem-duplication-correcting codes, we compare them with known tandem-duplication-correcting codes as well (see \cite{jf, fs, ze1, ze2, k, lw, gp, ys, ty, tf, tw,jf2,fs2, tw,tang noisy, tang isit}). The coding capacity of the channel allowing any number of 1-tandem duplications is $\log_{q}(q-1)$~\cite{jf,ys2}, which is slightly larger than our asymptotic code rate $\log_{q}(q-2)$. The gap is a consequence of introducing complement insertions into the noisy insertion channel.
\subsection{Decoding Algorithm and Complexity}
\label{sec:decoding}

In this subsection, we present an explicit decoding algorithm for the proposed codes $\mathcal{C}_{F(a,b,d,h,w,e,f,g)}(n)$ and establish its time complexity. 

Let $y \in D_{ct}^{*(\le 1)}(x)$ be the received sequence corresponding to the transmitted codeword $x \in \mathcal{C}_{F(a,b,d,h,w,e,f,g)}(n)$. By Lemma \ref{a noisy alter}, the random insertion alters the signature $\sigma(x)$ into $r = \sigma(y)$ through exactly one of three errors: a point insertion, a point substitution by a complementary symbol, or a burst insertion of length 2. Since the codebook is restricted to $\mathrm{Irr}(n)$, recovering the signature $\sigma(x)$ uniquely determines the codeword $x$. The decoding procedure is outlined in Algorithm \ref{alg:decoding}.

\begin{algorithm}[htbp]
	\caption{Decoding for $\mathcal{C}_{F(a,b,d,h,w,e,f,g)}(n)$}
	\label{alg:decoding}
	\begin{algorithmic}[1]
		\Require Received sequence $y \in D_{ct}^{*(\le 1)}(x)$
		\Ensure Transmitted codeword $x \in \mathcal{C}_{F(a,b,d,h,w,e,f,g)}(n)$
		\State Compute $r = \sigma(y)$ by collapsing consecutive identical and complementary symbols in $y$.
		\If{$|r| = n+1$}
		\State
		\State Decode $r$ using the modified VT decoder for $C_{T(a,d)}(n)$ to obtain $x$.
		\ElsIf{$|r| = n+2$}
		\State
		\State Form the $2 \times \frac{n+2}{2}$ array $A_2(r)$.
		\State Decode $A_2(r)$ using the SVT decoder for $C_{B(h,w,e,f,g)}(n)$ to obtain $x$.
		\ElsIf{$|r| = n$}
		\State Compute syndromes $S_0 = \sum_{j=1}^{n} r_j \pmod{2q}$ and $S_1 = \sum_{j=1}^{n} j r_j \pmod{qn}$.
		\If{$S_0 \equiv a \pmod{2q}$ \textbf{and} $S_1 \equiv b \pmod{qn}$}
		\State $x \leftarrow r$ 
		\Else
		\State 
		\State Determine the error value $v$ and position $i$ satisfying $q - 1 - 2v \equiv S_0 - a \pmod{2q}$ and $i(q - 1 - 2v) \equiv S_1 - b \pmod{qn}$.
		\State Set $x \leftarrow r$, and update $x_i \leftarrow q - 1 - v$.
		\EndIf
		\EndIf
		\State \Return $x$
	\end{algorithmic}
\end{algorithm}

\textbf{Complexity Analysis.} Let $N = |y|$ denote the length of the received sequence. Because the noisy insertion channel permits an arbitrary number of complement insertions and 1-tandem duplications, $N$ can be significantly larger than the codeword length $n$. Step 1 computes the signature $r = \sigma(y)$ via a single sequential scan of $y$, which takes $O(N)$ time. This operation effectively compresses the potentially unbounded sequence $y$ into a signature $r$ whose length is strictly bounded by $n+2$. The complexity of the subsequent algebraic recovery depends on the length of $r$. For $|r| \in \{n+1, n+2\}$, applying either the modified VT decoder or the SVT decoder takes $O(n)$ time. For $|r| = n$, evaluating the constraint syndromes $S_0$ and $S_1$ requires $O(n)$ arithmetic operations, and solving the resulting modular equations takes $O(1)$ time. In all cases, the algebraic decoding phase is strictly bounded by $O(n)$. Therefore, the overall time complexity of the algorithm is $O(N) + O(n) = O(N)$. This confirms that the decoding procedure operates in strictly linear time with respect to the length of the received sequence.


\section{Conclusion}\label{conclusion}
This paper established a coding framework for error correction in noisy DNA storage channels characterized by three dominant error types: arbitrary complement insertions, arbitrary 1-tandem duplications, and up to one random insertion. We determined the coding capacity of this channel to be $\mathrm{cap}_q^{\mathrm{noisy}} = \log_q(q-2)$ for alphabet size $q \geq 4$, demonstrating that correcting the additional random insertion error incurs no asymptotic rate penalty compared to channels limited to complement insertions alone. Through a code construction combining signature-preservation techniques with constrained coding frameworks, we developed asymptotically optimal codes that achieved this capacity.\par 
Future research directions include extending the model to multiple random insertions and designing codes for noisy channels in which complement insertions, tandem duplications of longer length and random insertions may occur. We propose the following two open problems for future investigation.
\begin{open problem}
Construct error-correcting codes for noisy insertion channel where arbitrary complement insertions, arbitrary 1-tandem duplications and up to $p$ random insertions may occur. Furthermore, we conjecture that the coding capacity of the channel remains $\log_{q}(q-2)$.
\end{open problem}
\begin{open problem}
Construct error-correcting codes for noisy insertion channel where arbitrary complement insertions, arbitrary tandem duplications of length at most 3 and up to $p$ random insertions may happen. 
\end{open problem}
The reader is invited to attack the above open problems.

	\section*{Acknowledgments}
The authors would like to thank Dr. Vidya Sagar for his suggestion and discussion.
	
\end{document}